\documentclass[showpacs,amsmath,amssymb,twocolumn,pra]{revtex4-1}

\usepackage[dvips]{graphicx} 
\usepackage{amsfonts}
\usepackage{amssymb}
\usepackage{amscd}
\usepackage{amsmath}    
\usepackage{enumerate}
\usepackage{epsfig}
\usepackage{subfigure}
\usepackage{bm}
\usepackage{xcolor}
\usepackage{amsthm}
\usepackage{framed}
\usepackage{multirow}
\usepackage{mathrsfs,amssymb}
\usepackage{latexsym} 
\usepackage{framed}
\usepackage{physics}
\usepackage{epstopdf}
\usepackage[skins]{tcolorbox}
\newtcolorbox[auto counter]{tbox}[2][]{%
	enhanced, float=hbt, drop fuzzy shadow southeast,
	colback=white!5!white, colframe=white!50!black,
	width= .97\columnwidth,sharp corners, boxrule=0.8pt,
	title={Table \thetcbcounter: #2}, #1
}
\usepackage[linesnumbered,boxed,ruled,vlined,rightnl]{algorithm2e}
\usepackage{qcircuit}

\newtheorem{lemma}{Lemma}
\newtheorem{proposition}{Proposition}

\begin{document}
\title{A hybrid algorithm simulating non-equilibrium steady states of an open quantum system}
\begin{abstract}
Non-equilibrium steady states are a focal point of research in the study of open quantum systems. Previous variational algorithms for searching these steady states have suffered from resource-intensive implementations due to vectorization or purification of the system density matrix, requiring large qubit resources and long-range coupling. In this work, we present a novel variational quantum algorithm that efficiently searches for non-equilibrium steady states by simulating the operator-sum form of the Lindblad equation.
By introducing the technique of random measurement, we are able to estimate the nonlinear cost function while reducing the required qubit resources by half compared to previous methods. Additionally, we prove the existence of the parameter shift rule in our variational algorithm, enabling efficient updates of circuit parameters using gradient-based classical algorithms.
To demonstrate the performance of our algorithm, we conduct simulations for dissipative quantum transverse Ising and Heisenberg models, achieving highly accurate results. Our approach offers a promising solution for effectively addressing non-equilibrium steady state problems while overcoming computational limitations and implementation challenges.
\end{abstract}

\author{Hongyi Zhou}
\email{zhouhongyi@ict.ac.cn}
\affiliation{State Key Lab of Processors, Institute of Computing Technology, Chinese Academy of Sciences, 100190, Beijing, China.}

\author{Rui Mao}
\affiliation{State Key Lab of Processors, Institute of Computing Technology, Chinese Academy of Sciences, 100190, Beijing, China.}

\author{Xiaoming Sun}
\affiliation{State Key Lab of Processors, Institute of Computing Technology, Chinese Academy of Sciences, 100190, Beijing, China.}

\maketitle

\section{Introduction}



Unlike ideal closed quantum systems, realistic quantum systems are inevitably coupled to the environment during their evolution, resulting in dissipation. These quantum systems are referred to as open quantum systems. The evolution of an open system cannot be described solely by the Schrödinger equation. Instead, it is usually described by the Lindblad master equation \cite{gorini1976completely}, assuming a memoryless environment where the correlation time of the environment is negligible compared to the characteristic timescale of system-environment interactions. A significant task in the study of open quantum systems is to search for the non-equilibrium steady states, which represent solutions to the Lindblad master equation. Non-equilibrium steady states have been extensively studied in statistical mechanics \cite{zia2007probability,freitas2022emergent}, biology \cite{sorrenti2017non}, and even quantum information processing tasks \cite{verstraete2009quantum,harrington2022engineered} due to their special transport properties. In quantum computation, dissipation engineering allows the preparation of graph states, which serve as resources for measurement-based quantum computation \cite{verstraete2009quantum}.

In the literature, various classical algorithms have been developed to search for non-equilibrium steady states. These include matrix product state and tensor network schemes \cite{cui2015variational,kshetrimayum2017simple}, real-space renormalization approaches \cite{finazzi2015corner}, and cluster mean-field approaches \cite{jin2016cluster}. For Hamiltonian systems, quantum Monte Carlo methods have been employed to stochastically sample system properties. Two promising approaches in this regard are versatile projector Monte Carlo techniques \cite{umrigar2015observations} and the variational Monte Carlo method \cite{ido2015time}. The latter can also be combined with neural network ansatz \cite{vicentini2019variational}.

Simulating large quantum systems using classical algorithms is challenging due to computational limitations. Quantum algorithms face implementation difficulties as well, mainly related to integrating a large number of qubits and mitigating decoherence effects. In contrast, hybrid quantum-classical algorithms, which only require shallow circuits without error correction, can be easily implemented on current noisy intermediate-scale quantum (NISQ) devices \cite{preskill2018quantum}. Variational quantum algorithms, an important subclass of hybrid quantum-classical algorithms, transform problems into circuit parameter optimization and updates, showing advantages in quantum chemistry \cite{mcardle2020quantum} and combinatorial optimization \cite{farhi2014quantum}.

Recently, some variational quantum algorithms for searching non-equilibrium steady states was proposed \cite{yoshioka2020variational,liu2021variational}. One way is to consider the vectorization form of the Lindblad master equation \cite{yoshioka2020variational}, where the density matrix of the open system is vectorized, and the system's evolution can be described by a single matrix called the Lindbladian, analogous to the Hamiltonian in the Schrödinger equation. The other way is purifying the open system with environment \cite{liu2021variational}.
By constructing appropriate cost functions, the steady state can be obtained using a process similar to the variational quantum eigensolver (VQE). However, both methods require at least twice the number of qubits of the system size and involve long-range coupling, making the physical implementations quite challenging. An open question is whether the number of qubits can be reduced while achieving the same task.

In our work, we propose a novel variational algorithm for searching non-equilibrium steady states. The chosen cost function is the squared Frobenius norm of the operator-sum form of Lindblad equation, which is a quadratic function of the density matrix of the mixed state. By performing random measurements at the end of the variational ansatz, we can estimate the quadratic function by using classical shadows of a single copy of the mixed state, requiring only half the number of qubits compared to \cite{yoshioka2020variational,liu2021variational}. We prove the existence of the parameter shift rule in the variational algorithm, enabling the application of gradient-based classical algorithms to update the circuit parameters. Finally, we present simulation results for a dissipative quantum transverse Ising model and a dissipative Heisenberg model, demonstrating the high accuracy achieved by our algorithm.

\section{Main result}

\subsection{Ansatz for preparing a mixed state}
The ansatz for preparing a mixed state is given in Fig.~\ref{fig:ansatz}.
A general $n$-qubit mixed state $\rho$ can be expressed as a spectral decomposition $\rho=\sum_{b=1}^{2^n} \lambda_b \ket{\psi_b}\bra{\psi_b}$. The idea is to generate the probability distribution $\{\lambda_b\}_b$ and the eigenstates $\{\ket{\psi_b}\}_b$, which are realized by a concatenation of a parameterized unitary gate $U_D(\vec{\theta}_D)$ and an intermediate measurement and another parameterized unitary gate $U_V(\vec{\theta}_V)$, respectively.
The first unitary $U_D(\vec{\theta}_D)$ is for eigenvalue distribution on computational basis $\{\ket{b}\}_b$, i.e., $\ket{D}=U_D(\vec{\theta}_D)\ket{0}^{\otimes n} =\sum_{q=1}^{2^n} \lambda_b \ket{b}$. Then the intermediate measurement will generate a probability distribution, $p_b = |\lambda_b|^2$. One possible implementation $U_D(\vec{\theta}_D)$ of is given in the left of Fig.~\ref{fig:ansatz1}. To achieve a better expressibility, this parameterized circuit can be repeated for several times with independent parameters in each repetition. The number of repetitions is denoted as $D_1$.
The second unitary realizes the basis rotation, after which the state becomes a general mixed state $\rho = \sum_{b=1}^{2^n} p_b \ket{\psi_b}\bra{\psi_b}$, where $\ket{\psi_b} = U_V(\vec{\theta}_V)\ket{b}$. The concrete implementation of $U_V(\vec{\theta}_V)$ is the hardware-efficient ansatz following \cite{yoshioka2020variational}, which is given in the right of Fig.~\ref{fig:ansatz1}. The gates in the dashed box may also be repeated for a better expressibility. The number of repetitions is denoted as $D_2$. The circuit parameters are jointly denoted as $\vec{\theta}=(\vec{\theta}_D,\vec{\theta}_V)$.

At the end of the circuit, a random measurement is performed on the output state $\rho$, which is realized by a random unitary $U_{\mathrm{rand}}$ followed by a projective measurement on the computational basis.
Compared with the ansatz in \cite{PhysRevResearch.2.043289}, the number of qubits is reduced from $2n$ to $n$.
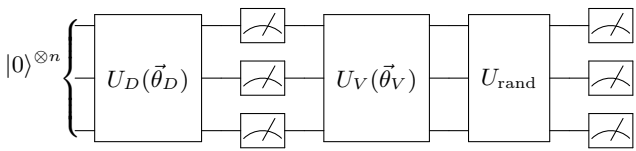
\begin{figure}[h!]
	\centerline{
		\Qcircuit @C = 0.8 em @R =0.75 em  {
				\lstick{}& \multigate{2}{U_D(\vec{\theta}_D)} &\qw        & \meter   &\qw   & \multigate{2}{U_V(\vec{\theta}_V)} & \qw  & \multigate{2}{U_{\rm{rand}}} & \qw & \meter\\
				\lstick{}& \ghost{U_D(\vec{\theta}_D)}        &\qw        & \meter   &\qw   & \ghost{U_V(\vec{\theta}_V)}        & \qw  & \ghost{U_{\rm{rand}}}     & \qw & \meter\\
				\lstick{}& \ghost{U_D(\vec{\theta}_D)}        &\qw        & \meter   &\qw   & \ghost{U_V(\vec{\theta}_V)}        & \qw  & \ghost{U_{\rm{rand}}}     & \qw & \meter
				\inputgroupv{1}{3}{0.4em}{1.6em}{\ket{0}^{\otimes n}}
			}
	}
	\caption{Quantum circuit for the mixed state preparation and random measurements.}
	\label{fig:ansatz}
\end{figure}

\begin{figure*}[!ht]
    \begin{minipage}[t]{0.49\textwidth}
    \centerline{
        \Qcircuit @C = 0.8 em @R =0.75 em  {
                \lstick{}& \gate{R_y} & \ctrl{1}
                &\qw &\qw
                \\
                \lstick{}& \gate{R_y} & \gate{R_y} 
               & \ctrl{1} &\qw
               \\ 
                \lstick{}& \gate{R_y} & \qw 
                & \gate{R_y}  &\qw
               \gategroup{1}{2}{3}{4}{1em}{--}
            }
    }
    \end{minipage}
    \begin{minipage}[t]{0.49\textwidth}
    \centerline{
        \Qcircuit @C = 0.8 em @R =0.75 em  {
                & \lstick{}& \gate{R_y} & \gate{R_z} 
                &\ctrl{1} & \qw  &\qw  &\gate{R_y}              & \gate{R_z} & \qw 
                \\
               & \lstick{}& \gate{R_y} & \gate{R_z} 
               & \control \qw & \ctrl{1}       &\qw     &\gate{R_y}               & \gate{R_z} & \qw 
               \\ 
            & \lstick{}& \gate{R_y} & \gate{R_z} 
                & \qw &  \control \qw &       \qw &       \gate{R_y}             & \gate{R_z} &\qw 
               \gategroup{1}{3}{3}{6}{1em}{--}
            }
    }
    \end{minipage}
    \caption{The ansatz for eigenvalue distribution $U_D(\vec{\theta}_D)$ (left) and basis rotation $U_V(\vec{\theta}_V)$ (right). Each gate is parameterized by a single parameter. For a better expressibility, the structures in the dashed boxed can be repeated for $D_1$ and $D_2$ times, respectively.} 
    \label{fig:ansatz1}
\end{figure*}
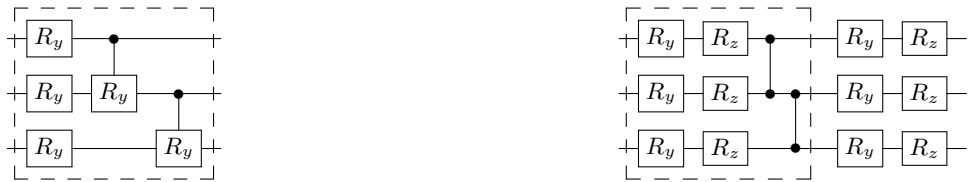

\subsection{Cost function}
The evolution of an open system coupled with a memoryless environment follows the Lindblad master equation,
\begin{equation}\label{eq:lindblad}
	\frac{d\rho}{dt}={\cal L}(\rho) = -\mathrm{i}[H,\rho] + \sum_{k=1}^K \frac{\gamma_k}{2}\left(2c_k \rho c_k^\dag - \{c_kc_k^\dag, \rho\} \right),
\end{equation}
where $c_k$ is the $k$-th jump operator that determines the dissipation and $\gamma_k$ is the strength of the dissipation. The steady states $\rho_{\mathrm{ss}}$ will satisfy
\begin{equation}
	{\cal L}(\rho_{\mathrm{ss}}) = 0.
\end{equation}
The search for the steady states can be transformed into an optimization problem. The cost function $C(\vec{\theta})$ should satisfy
\begin{equation}\label{eq:condition_costfunction}
	\begin{aligned}
		C(\vec{\theta}) & \geq 0, \quad \forall \vec{\theta}                                              \\ 
		C(\vec{\theta}) & =0 , \quad \text{iff.}  \,\, {\cal L}(\rho(\vec{\theta}))=0. \\
	\end{aligned}
\end{equation}
Then an arbitrary matrix norm will satisfy the conditions in Eq.~\eqref{eq:condition_costfunction}. We choose the square of the Frobenius norm $||A||_F^2$ where $||A||_F =\sqrt{\mathrm{tr}(A^\dag A)}$ as the cost function,
\begin{equation}\label{eq:def_costfunction}
	\begin{aligned}
		C(\vec{\theta}) := ||{\cal L}(\rho)||_F^2.
	\end{aligned}
\end{equation}
Then the cost function in Eq.~\eqref{eq:def_costfunction} satisfies Eq.~\eqref{eq:condition_costfunction} by definition. In the variational algorithm, the searching process is to minimize Eq.~\eqref{eq:def_costfunction} iteratively until the cost function converges to zero. Then the steady state is found.

\subsection{Calculating the cost function}

The cost function in Eq.~\eqref{eq:def_costfunction} is a quadratic function of a $n$-qubit quantum state $\rho$, which is conventionally obtained by preparing two copies $\rho\otimes \rho$ and performing the swap test.
Here we consider an alternative approach of shadow tomography on only one copy of $\rho$. In shadow tomography, one performs random measurement and obtain a number of classical snapshots of $\rho$, denoted as \cite{Huang2020PredictingMP}
\begin{equation}
	\hat{\rho} = {\cal M}^{-1} \left(U^\dag \ket{b}\bra{b} U\right),
\end{equation}
where $U$ is a random unitary chosen from an ensemble ${\cal U}$, $\ket{b}\bra{b}$ is a projective measurement on the computational basis, and ${\cal M}^{-1}$ is the inverse channel (realized by post-processing) depending on the ensemble ${\cal U}$ (typically random Pauli gates or Clifford gates). Then the expectation is calculated by replacing $\rho$ with its snapshots $\hat{\rho}$ and take an expectation over all possible $U$ and $\ket{b}$.
In our case, we only care about quadratic functions of $\rho$ in the form of $\mathrm{tr}(O_1 \rho O_2 \rho)$. We consider the empirical mean value
\begin{equation}\label{eq:ube}
	\begin{aligned}
		\frac{1}{N(N-1)} \sum_{i \neq {i^\prime}}  \mathrm{tr}(O_1 \hat{\rho}_i O_2 \hat{\rho}_{i^\prime}),
	\end{aligned}
\end{equation}
which $N$ is the number of random unitary gates, $i$ is the label of the unitary gate, and $\hat{\rho}_i$ is the classical shadow defined in Eq.~\eqref{eq:classicalshadow}. Then Eq.~\eqref{eq:ube} is an unbiased estimator of $\mathrm{tr}(O_1 \rho O_2 \rho)$, i.e.,
\begin{equation}
	\mathbb{E}\left[\frac{1}{N(N-1)} \sum_{i \neq {i^\prime}}  \mathrm{tr}(O_1 \hat{\rho}_i O_2 \hat{\rho}_{i^\prime})\right] = \mathrm{tr}(O_1 \rho O_2 \rho),
\end{equation}
where the expectation without subscripts is taken over all unitary gates in $\cal{U}$ and measurement outcomes throughout the paper.
We refer to Appendix~\ref{app:ube} for the proof.

The cost function is a linear combination of terms in the form of $\mathrm{tr}(O_1 \rho O_2 \rho)$, which can be estimated by
\begin{widetext}
	\begin{equation}\label{shadow_estimation}
		\begin{aligned}
			C(\vec{\theta}) & = \frac{1}{N(N-1)}\sum_{i\neq {i^\prime}} \mathrm{tr}\left(-\mathrm{i}H\hat{\rho}_{i}+ \mathrm{i}\hat{\rho}_{i} H +\sum_{k=1}^K \frac{\gamma_k}{2}(2c_k \hat{\rho}_{i} c_k^\dag-c_kc_k^\dag\hat{\rho}_{i}-\hat{\rho}_{i} c_kc_k^\dag) \right) \\
			                & \times \left(\mathrm{i}H\hat{\rho}_{i^\prime}-\mathrm{i}\hat{\rho}_{i^\prime} H +\sum_{k=1}^K \frac{\gamma_k}{2}(2c_k^\dag \hat{\rho}_{i^\prime} c_k-\hat{\rho}_{j} c_kc_k^\dag-c_kc_k^\dag\hat{\rho}_{i^\prime})\right).
		\end{aligned}
	\end{equation}
\end{widetext}
The operators $O_1$ and $O_2$ in the shadow estimation are typically Hermitian. While we notice that in Eq.~\eqref{shadow_estimation}, $O_1$ and $O_2$ are not Hermitian for the terms with $c_k \hat{\rho} c_k^\dag$ and $c_k^\dag \hat{\rho} c_k$. To estimate these terms, we apply the following lemma.
\begin{lemma}
	Suppose $\Phi$ is a linear map on the set of density matrix, $\Phi(\rho) = \sigma$, and $H_1$ and $H_2$ are two Hermitian operators. Then $\mathbb{E}\left(\mathrm{tr}(H_1 \Phi(\hat{\rho}_i)H_2 \hat{\rho}_{i^\prime}) \right)= \mathrm{tr}(H_1 \sigma H_2 \rho) $.
\end{lemma}
\begin{proof}
	We have the following relations
	\begin{equation}
		\begin{aligned}
			\mathbb{E}\left(\mathrm{tr}(H_1 \Phi(\hat{\rho}_i)H_2 \hat{\rho}_j) \right) & =\mathrm{tr}\left(H_1 \Phi(\mathbb{E}(\hat{\rho}_i) ) H_2 \mathbb{E}(\hat{\rho}_{i^\prime})\right) \\
	     & = \mathrm{tr}\left(H_1 \Phi(\rho) H_2 \rho \right)                                                 \\
 & = \mathrm{tr}\left(H_1 \sigma H_2 \rho \right),
		\end{aligned}
	\end{equation}
	where the first equality comes from the expectation of independent random variables. The classical shadows $\hat{\rho}_i$ and $\hat{\rho}_{i^\prime}$ are independent as long as $i \neq j$.
\end{proof}
With this lemma, we can deal with $c_k \rho c_k^\dag$ in Eq.~\eqref{eq:lindblad} as a linear map of $\rho$. For the terms where $O_1$ and $O_2$ are not Hermitian, one can still calculate $\mathrm{tr}(O_1 \hat{\rho}_i O_2 \hat{\rho}_{i^\prime})$, which is an unbiased estimator for 
$\mathrm{tr}(O_1 \rho O_2 \rho)$.

For an arbitrary physical observable $O$, the expectation $\mathrm{tr}(\rho O)$ can be directly estimated with the classical shadows of $\rho$. The number of measurements is estimated in Appendix~\ref{app:number}. This simplifies the measurement process in \cite{yoshioka2020variational} where a stochastic sampling is required to estimate $p_b$ and the observable is individually measured in each $\ket{\psi_b}$.
The shadow estimation can be further combined with the virtual distillation technique during the optimization process, which is discussed in Appendix~\ref{app:virtualdistillation} in detail.

\subsection{Updating parameters}
Gradient-based methods are widely applied to updating parameters in variational algorithms. A widely applied approach to calculate the gradients is the parameter-shift rule. Suppose each unitary gate $U_i$ is parameterized by a single parameter $\theta_i\in \vec{\theta}$, i.e., the whole unitary operation is given by $U=U_n(\theta_n)U_{n-1}(\theta_{n-1})\cdots U_1(\theta_1)$. We prove that the gradient of the cost function $C(\vec{\theta})$ can be calculate by the parameter-shift rule.


\begin{lemma}\label{lem:param_shift_iff}
	Suppose the cost function $C$ can be represented as a Fourier integral
	$\int _{\mathbb{R}{^{l}}}\hat{C}( \vec{\omega} ) e^{\mathrm{i}\vec{\omega} \cdot \vec{\theta} } d\vec{\theta} $,
	then there exists a "parameter shift rule" with respect to fixed set of "shift"s $S$ and fixed coefficients $\{c _{\vec{\Delta} } |\vec{\Delta} \in S\}$
	\begin{equation}\label{eq:param_shift}
		C'( \vec{\theta} ) =\sum _{\vec{\Delta} \in S} c _{\vec{\Delta} } C( \vec{\theta} +\vec{\Delta} )
	\end{equation}
	iff. $C$ is a linear combination of finite trigonometric functions
	\begin{equation}\label{eq:trigon_expand}
		C( \vec{\theta} ) =\sum _{1\leqslant k\leqslant L} \beta _{k} e^{\mathrm{i}\vec{\gamma}_k \cdot \vec{\theta} }
	\end{equation}
	where $\beta _{k} \in \mathbb{C}$ and $\vec{\gamma}_k \in \mathbb{R}^{l}$.
\end{lemma}

\begin{proof}
	For simplicity we prove this lemma for single-variable cost function $C$ where $l=1$ and $\theta $ is a scalar. The case of multi-variables is straightforward. We first design a parameter shift rule for $C=\sum _{1\leqslant k\leqslant L} \beta _{k} e^{\mathrm{i}\gamma _{k} \theta }$. It is easy to check that Vandermonde matrix $\Delta _{L}\left( e^{\mathrm{i}\alpha \gamma _{1}} ,e^{\mathrm{i}\alpha \gamma _{2}} ,\dotsc ,e^{\mathrm{i}\alpha \gamma _{L}}\right)$ transforms $\left( \beta _{1} e^{\mathrm{i}\gamma _{1} \theta } ,\beta _{2} e^{\mathrm{i}\gamma _{2} \theta }, \dotsc ,\beta _{L} e^{\mathrm{i}\gamma_{L} \theta }\right)^{T}$ to $( C( \theta ), C( \theta +\alpha ) ,\dotsc ,C( \theta +( L-1) \alpha ))^{T}$:
	\begin{equation}
		\begin{pmatrix}
			C( \theta )         \\
			C( \theta +\alpha ) \\
			\vdots              \\
			C( \theta +( L-1) \alpha )
		\end{pmatrix} =
		\Delta_L \cdot
		\begin{pmatrix}
			\beta _{1} e^{\mathrm{i}\gamma _{1} \theta } \\
			\beta _{2} e^{\mathrm{i}\gamma _{2} \theta } \\
			\vdots                              \\
			\beta _{L} e^{\mathrm{i}\gamma _{L} \theta }
		\end{pmatrix}
	\end{equation}
	where
	\begin{equation}
		\Delta_L :=
		\begin{pmatrix}
			1                              & 1                              & \cdots & 1                              \\
			e^{\mathrm{i}\alpha \gamma _{1}}        & e^{\mathrm{i}\alpha \gamma _{2}}        & \cdots & e^{\mathrm{i}\alpha \gamma _{L}}        \\
			\vdots                         & \vdots                         & \ddots & \vdots                         \\
			e^{\mathrm{i}( L-1) \alpha \gamma _{1}} & e^{\mathrm{i}( L-1) \alpha \gamma _{2}} & \cdots & e^{\mathrm{i}( L-1) \alpha \gamma _{L}}
		\end{pmatrix}
	\end{equation}
	Since $\gamma _{p} \neq \gamma _{q}$ for each pair of $p\neq q$, there exists $\alpha \in \mathbb{R}$ such that $e^{\mathrm{i}\alpha \gamma _{p}} \neq e^{\mathrm{i}\alpha \gamma _{q}}$ for each pair of $p\neq q$. In this case, the Vandermonde matrix $\Delta _{L}$ is invertible and it follows immediately that
	\begin{equation}
		\begin{aligned}
			C'( \theta )
			 & =\sum _{1\leqslant k\leqslant L} \mathrm{i}\gamma _{k} \beta _{k} e^{\mathrm{i}\gamma _{k} \theta } \\
			 & =\mathrm{i}\begin{pmatrix}
				     \gamma _{1} & \gamma _{2} & \cdots & \gamma _{L}
			     \end{pmatrix} \cdot \Delta _{L}^{-1} \cdot
			\begin{pmatrix}
				C( \theta )         \\
				C( \theta +\alpha ) \\
				\vdots              \\
				C( \theta +( L-1) \alpha )
			\end{pmatrix}
		\end{aligned}
	\end{equation}
	Note that $\gamma _{k} ,\alpha $ are independent of $\theta $.

	For the necessary side, we use the properties of Fourier transform. Performing Fourier transform on both sides of (\ref{eq:param_shift}) gives
	\begin{equation}
		\mathrm{i}\omega \hat{C} (\omega )=\sum _{\Delta \in S} c_{\Delta } e^{-\mathrm{i} \omega \Delta }\hat{C} (\omega ),\quad \omega \in \mathbb{R}
	\end{equation}
	Notice that the function $f( \omega ) =i\omega -\sum _{\Delta \in S} c _{\Delta } e^{-\mathrm{i}\omega \Delta }$ is analytic and non-constant, thus by identity theorem it has finite zero points. Suppose the number of zero points of $f$ is $L^\prime$. Then the fact that $f( \omega )\hat{C} (\omega )=0$ always holds implies that $\hat{C} (\omega )\neq 0$ at $L\leq L^\prime$ points, i.e., there exists $\gamma _{k} $ and $\beta _{k}$ $(1\leqslant k\leqslant L)$, such that
	\begin{equation}
		\hat{C} (\omega )=\sum _{1\leqslant k\leqslant L} \beta _{k} \delta (\omega -\gamma _{k} ).
	\end{equation}
	Perform the inverse Fourier transform we get (\ref{eq:trigon_expand}).
\end{proof}

\begin{proposition}[Parameter shift rule for $C$]
	Suppose a parameterized ansatz consists of Pauli rotations and unparameterized gates, 
 \begin{equation}
 \begin{aligned}
 U_D(\vec{\theta}_{D}) &=W_{0}\left[\prod _{j}\exp\left( -\mathrm{i}\frac{\theta _{j}}{2} P_{j}\right) W_{j}\right] \\ 
 U_V(\vec{\theta}_{V}) &=W'_{0}\left[\prod _{k}\exp\left( -i\frac{\theta_{k}}{2} P_{k}\right) W'_{k}\right], 
 \end{aligned}
 \end{equation}
 where $P_j$ is a Pauli operator $P_j\in \{I, \sigma_x, \sigma_y, \sigma_z\}$, $W_j$ is an unparameterized gate, and $\theta _{j}$ are mutually independent for different $j$. Then
	\begin{equation}\label{eq:parametershift}
		\begin{aligned}
			\partial _{\theta _{j}} C(\vec{\theta })
			 & =
			\left[ C\left(\vec{\theta } +\frac{\pi }{4}\vec{e}_{j}\right) -C\left(\vec{\theta } -\frac{\pi }{4}\vec{e }_{j}\right)\right]                         \\
			 & -\frac{\sqrt{2} -1}{2}\left[ C\left(\vec{\theta } +\frac{\pi }{2}\vec{e}_{j}\right) -C\left(\vec{\theta } -\frac{\pi }{2}\vec{e}_{j}\right)\right]
		\end{aligned}
	\end{equation}
	where $\vec{e}_{j}$ is the $j$th unit vector.
\end{proposition}

\begin{proof}
	Recall that
	\begin{equation}
		\begin{aligned}
			\rho (\vec{\theta }) & =U_V(\vec{\theta }_{V})\Big(\sum _{b} \bra{b}U_D(\vec{\theta }_{D}) \op{0}U_D^{\dagger }(\vec{\theta }_{D}) \ket{b}\op{b}\Big) U_V^{\dagger }(\vec{\theta }_{V}) , \\
			C(\vec{\theta })     & =\mathrm{tr}\left(\mathcal{L( \rho (\vec{\theta })) L}( \rho (\vec{\theta }))^{\dagger }\right) .
		\end{aligned}
	\end{equation}
	Since each parameterized gates are Pauli rotations with eigenvalues $\pm 1$, such gates admit the following decomposition:
	\begin{equation}
		\exp\left( -\mathrm{i}\frac{\theta _{j}}{2} P_{j}\right) =e^{\mathrm{i}\frac{\theta_j }{2}} A+e^{-\mathrm{i}\frac{\theta_j }{2}} B,
	\end{equation}
	for some unparameterized matrices $A$ and $B$ .
	It is easy to check that there exist constants $\beta_k$, such that
	\begin{equation}
		C( \vec{\theta} ) =\sum _{1\leqslant k\leqslant L} \beta _{k} e^{\mathrm{i}\vec{\gamma}_k \cdot \vec{\theta} },\quad \vec{\gamma}_k\in \{0,\pm 1,\pm 2\}^l.
	\end{equation}
	Similiar to the proof of the sufficient side of Lemma \ref{lem:param_shift_iff}, we have
	\begin{align*}
		\partial _{\theta _{j}} C
		 & =\mathrm{i}\begin{pmatrix}
			     \gamma _{1} & \dotsc & \gamma _{5}
		     \end{pmatrix}
		\cdot \Delta ^{-1} \cdot
		\begin{pmatrix}
			C(\vec{\theta } +\alpha _{1}\vec{e}_{j}) \\
			\vdots                                   \\
			C(\vec{\theta } +\alpha _{5}\vec{e}_{j})
		\end{pmatrix}                           \\
		 & = \left[ C\left(\vec{\theta } +\frac{\pi }{4}\vec{e}_{j}\right)
		-C\left(\vec{\theta } -\frac{\pi }{4}\vec{e }_{j}\right)\right]    \\
		 & -\frac{\sqrt{2} -1}{2}\left[
			C\left(\vec{\theta } +\frac{\pi }{2}\vec{e}_{j}\right)
			-C\left(\vec{\theta } -\frac{\pi }{2}\vec{e}_{j}\right)\right]
	\end{align*}
	where $\vec{\gamma } =( -2,-1,0,1,2) ,\vec{\alpha } =\frac{\pi }{4}\vec{\gamma } ,\Delta =\left( e^{\mathrm{i}\alpha _{j} \gamma _{k}}\right)_{jk}$.
\end{proof}

\subsection{Algorithm}
We summarize the hybrid algorithm in Algorithm~\ref{alg:steadystate}. In the quantum part, random measurement is performed to obtain classical shadows, based on which we can estimate the cost function and its gradient. In the classical part, gradient-based algorithms such as gradient descent and BFGS \cite{fletcher2000practical} can be applied to update the ansatz parameters. After the classical optimization terminates, the ansatz parameters are denoted as $\vec{\theta}^*$. Then the quantum state $\rho(\vec{\theta}^*)$ is the steady state we are searching for. As mentioned above, to estimate the expectation of an arbitrary observable on the steady state $\mathrm{tr}(O\rho(\vec{\theta}^*))$, we can apply the classical shadows of $\rho(\vec{\theta}^*)$ obtained in the optimization process.
\begin{algorithm}[hbt] \caption{Search for steady state} \label{alg:steadystate}
	\KwIn{initial point $\vec{\theta}^{(0)}$; termination parameter $\epsilon$; initial value of the cost function $C(\vec{\theta}^{(0)}) = C_0 \gg \epsilon$}

	\textbf{While} $C(\vec{\theta}) \geq \epsilon$ \textbf{do} \\
	\quad update $\vec{\theta}$ by some gradient-based classical algorithm \\
	\quad \textbf{for} $i=1:N$ \textbf{do} \\
	\quad\quad Randomly choose a unitary $U_i$ from an ensemble $\mathcal{U}$ and record it, then perform the unitary.\\
	\quad\quad  \textbf{for} $j=1:M$ \textbf{do}  \\
	\quad\quad\quad Measure the output state with $\sigma_z^{\otimes n}$. Record the outcome as $\ket{b_i^j}$. \\
	\quad\quad \textbf{end for} \\
	\quad\quad Take the average of the $M$ measurement results
	\begin{equation}\label{eq:classicalshadow}
		\hat{\rho}_i = \frac{1}{M}\sum_{j=1}^M \mathcal{M}^{-1}\left(U_i^\dag \ket{b_i^j}\bra{b_i^j}  U_i\right).
	\end{equation}
	\quad \textbf{end for} \\
	\quad Get the shadow set $\{\hat{\rho}_1,\hat{\rho}_2,\dotsc, \hat{\rho}_N\}$. Estimate the cost function $C(\vec{\theta})$ by Eq.~\eqref{shadow_estimation}. Calculate the gradient of the cost function by Eq.~\eqref{eq:parametershift} \\
	\textbf{end While}   \\
\end{algorithm}

\section{Example}
\subsection{Ising Model}
We consider the Ising model with transverse-field, whose Hamiltonian is given by
\begin{equation}
	H = \frac{1}{2}\sum_i \sigma_i^z \sigma_{i+1}^z + g \sum_i \sigma_i^x,
\end{equation}
where $\sigma_i^z$ is the Pauli operator of the $i$-th spin and $g$ is the amplitude of the transverse field. The jump operators are
\begin{equation}
	\begin{aligned}
		c_{1,i} & = \sigma_i^-  \\
		c_{2,i} & = \sigma_i^z,
	\end{aligned}
\end{equation}
which characterize the damping and dephasing effect.
In our simulation, we choose the particle number $n=2$ with the dissipation strengths $\gamma_1 =1$ and $\gamma_2 = 0.5$. We set the repetition numbers $D_1=D_2=4$. The simulation results are shown in Fig.~\ref{fig:Ising}. To characterize the accuracy of the algorithm, we use infidelity defined as $1-(\mathrm{tr}(\sqrt{\rho_{\mathrm{e}}}\rho(\vec{\theta^*})\sqrt{\rho_{\mathrm{e}}})^{1/2})^2$, where $\rho_{\mathrm{e}}$ is the exact solution. We can see that the infidelity between $\rho(\vec{\theta^*})$ and $\rho_{\mathrm{e}}$ is always lower than $10^{-4}$. Such low infidelity implies a high accuracy of our algorithm. We also plot how the expectations of joint observables $XY$, $YY$ and $ZZ$ vary with the amplitude of the transverse field $g$.

\begin{figure}
	\centering
	\includegraphics[width=0.5\textwidth]{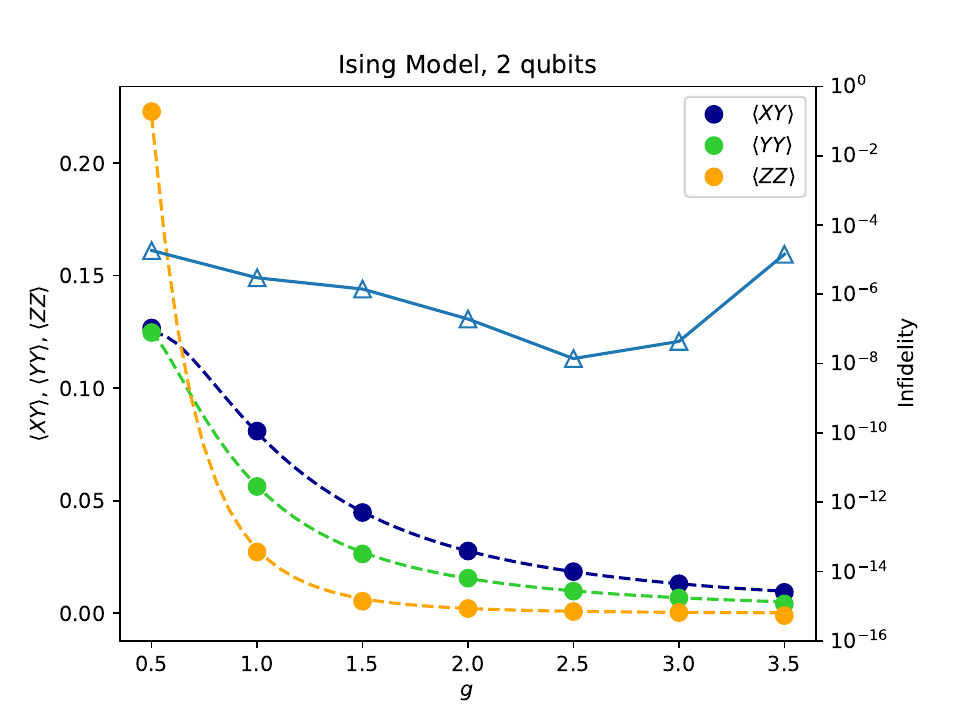}
	\caption{Simulation results for the dissipative Ising model.}
	\label{fig:Ising}
\end{figure}

\begin{figure}
	\centering
	\includegraphics[width=0.5\textwidth]{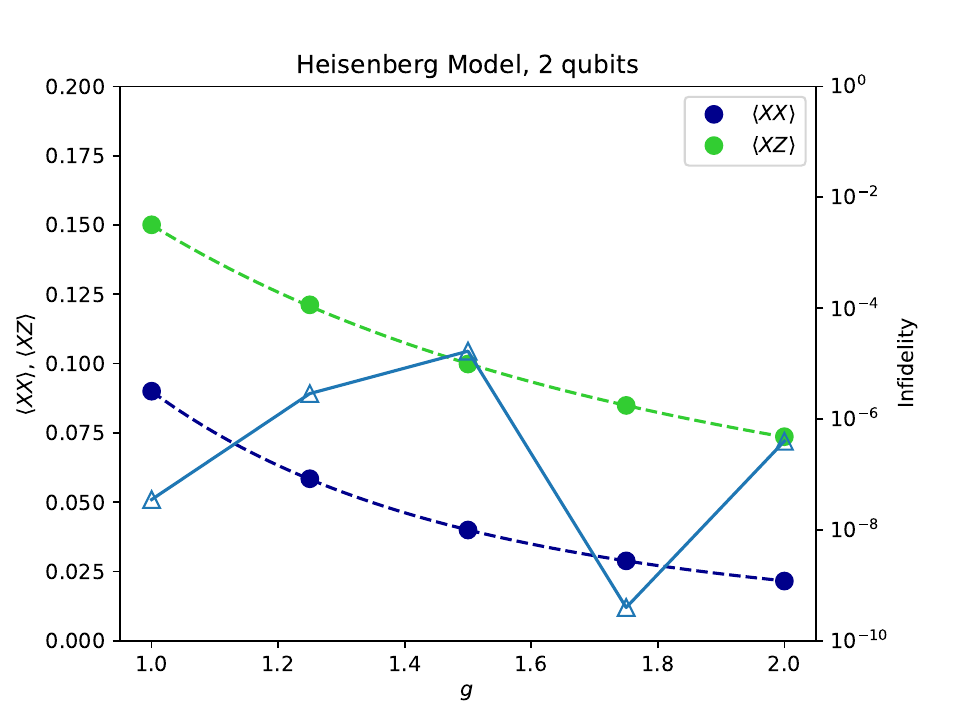}
	\caption{Simulation results for the dissipative Heisenberg model.}
	\label{fig:Heisenberg}
\end{figure}

\subsection{Heisenberg Model}
We also consider the Heisenberg model with a transverse-field in one direction, whose Hamiltonian is given by
\begin{equation}
	H = \sum_i \left(\sigma_i^z \sigma_{i+1}^z + \sigma_i^y \sigma_{i+1}^y + \sigma_i^x \sigma_{i+1}^x\right) + g \sum_i \sigma_i^x
\end{equation}
The jump operators are
\begin{equation}
	\begin{aligned}
		c_{1,i} & = \sigma_{i,z}^- =  \frac{1}{2}(\sigma_x - i\sigma_y)   \\
		c_{2,i} & = \sigma_{i,y}^- =  \frac{1}{2}(\sigma_x - i\sigma_z)   \\
		c_{3,i} & = \sigma_{i,x}^- =  \frac{1}{2}(-\sigma_z - i\sigma_y), \\
	\end{aligned}
\end{equation}
which characterize damping in three directions.

In our simulation, we choose the particle number $n=2$ with the dissipation strengths $\gamma_1 =1$ and set $D_1=D_2=2$. The simulation results are shown in Fig.~\ref{fig:Heisenberg}. We can see that the infideity is also lower than $10^{-4}$. The expectations of $XX$ and $XZ$ varying with $g$ are also plotted.

\section{conclusion}
In conclusion, we have presented a novel variational algorithm for searching non-equilibrium steady states in open quantum systems. By directly simulating the operator-sum form of the Lindblad equation, we successfully reduce the required number of qubits by half compared to previous approaches that utilize vectorization or purification. Through random measurements and parameter-shift rule, we are able to estimate the quadratic cost function and its gradient.

For future research, we can explore the extension of our algorithm to larger and more complex quantum systems, as well as investigate its applicability in other dissipative models. The random measurement technique may also be exploited to estimate overlaps between parameterized quantum states, which is the crucial step of the variational real and imaginary time simulation of open quantum systems.

\emph{Note added} When preparing this manuscript, we notice another related work \cite{PhysRevLett.130.240601} published very recently, where the authors also consider the operator-sum form of the Lindblad equation and prepare a mixed state. The difference is that the mixed state is a mixture over heuristically chosen ansatz states without parameterization. The performance of their algorithm relies on the choice of the ansatz states.

\section*{Acknowledgement}
This work was supported in part by the National Natural Science Foundation of China Grants No. 61832003, 62272441, 12204489, and the Strategic Priority Research Program of Chinese Academy of Sciences Grant No. XDB28000000.

\onecolumngrid
\appendix

\section{Proof of the unbiased estimator}\label{app:ube}
In this section we prove that Eq.~\eqref{eq:ube} is an unbiased estimator for $\mathrm{tr}(O_1\rho O_2 \rho)$. Recall that the classical shadow is given by
\begin{equation}
	\hat{\rho}_i = \frac{1}{M}\sum_{j=1}^M \hat{\rho}_{i,j}
\end{equation}
where
\begin{equation}
	\hat{\rho}_{i,j} =\mathcal{M}^{-1}\left(U_i^\dag \ket{b_i^j}\bra{b_i^j}  U_i\right)
\end{equation}
and $\mathbb{E}_{U_i,\ket{b_i^j}}(\hat{\rho}_{i,j}) = \rho$.
Then we can calculate the expectation
\begin{equation}
	\begin{aligned}
		\mathbb{E}_{U_i,\ket{b_i^j}, U_{i^\prime}, \ket{b_{i^\prime}^{j^\prime}} } \frac{1}{N(N-1)} \sum_{i \neq {i^\prime}}  \mathrm{tr}(O_1 \hat{\rho}_i O_2 \hat{\rho}_{i^\prime}) & =
		\mathbb{E}_{U_i,\ket{b_i^j}, U_{i^\prime}, \ket{b_{i^\prime}^{j^\prime}} }\left(\frac{1}{N(N-1)} \frac{1}{M^2} \sum_{i,i^\prime\neq i,j,j^\prime} \mathrm{tr}(O_1 \hat{\rho}_{i,j} O_2 \hat{\rho}_{i^\prime,j^\prime})\right)                                           \\
  & = \frac{1}{N(N-1)M^2}\left[\sum_{i,i^\prime\neq i,j,j^\prime} \mathrm{tr}\left(O_1 \mathbb{E}_{U_i,\ket{b_i^j} }(\hat{\rho}_{i,j}) O_2 \mathbb{E}_{U_{i^\prime}, \ket{b_{i^\prime}^{j^\prime}} }(\hat{\rho}_{i^\prime,j^\prime}) \right) \right] \\
	     & = \frac{1}{N(N-1)M^2}\left[N(N-1)M^2 \mathrm{tr}(O_1 \rho O_2 \rho) \right] \\
	        & = \mathrm{tr}(O_1 \rho O_2 \rho).
	\end{aligned}
\end{equation}
The second equation is due to the fact that $\hat{\rho}_{i,j}$ and $\hat{\rho}_{i^\prime,j^\prime}$ are independent random variables.
In the second last equation, the number of summation is $N(N-1)M^2$ for all $i\neq i^\prime$.

\section{Calculating the number of measurements}\label{app:number}

The variance bound of estimation of $\operatorname{tr}( O\rho \otimes \rho )$ by shadow tomography has been analyzed in \cite{Huang2020PredictingMP}. Specifically, the variance is upper bounded by $\frac{4x^{2}}{N^{2}} +\frac{4x}{N}$, where $x=4^{k} \| O\| _{\infty }^{2}$ when using Pauli measurements (where $k$ is the locality of $O$), and $x=\sqrt{9+6/2^{n}}\operatorname{tr}\left( O^{2}\right)$ when using Clifford measurements. In shadow estimation,
\begin{equation*}
	O=\sum _{ij}\left( B_{i}^{\dagger } B_{j} \otimes A_{j}^{\dagger } A_{i}\right) S,
\end{equation*}
where we have written the Lindblad equation in the form of $\mathcal{L} \rho =\sum _{i} A_{i} \rho B_{i}^{\dagger }$. By Chebyshev inequality, in order to achieve additive error $\varepsilon $ with success probability $1-\delta $, the number of measurements should scale as
\begin{equation*}
	N=\Omega \left(\frac{4^{k} \| O\| _{\infty }^{2}}{\varepsilon ^{2} \delta }\right),
\end{equation*}
when using Pauli measurements, and
\begin{equation*}
	N=\Omega \left(\frac{\sqrt{9+6/2^{n}}\operatorname{tr}\left( O^{2}\right)}{\varepsilon ^{2} \delta }\right),
\end{equation*}
when using Clifford measurements.


\section{Error mitigation}\label{app:virtualdistillation}
The error mitigation method is based on shadow distillation \cite{seif2023shadow}. It combines the virtual distillation technique \cite{huggins2021virtual} with classical shadows, which can also reduces the error exponentially with the number of copies without actually preparing the multiple copies of quantum states.
In the shadow distillation, one need to estimate expectations $\mathrm{tr}(O_1 \rho^m O_2 \rho^m)$ and $\mathrm{tr}(\rho^{m}\otimes \rho^m)$,
\begin{equation}\label{eq:highorderexpec}
	\begin{aligned}
		\mathrm{tr}(O_1 \rho^m O_2 \rho^m)  & = \mathrm{tr}\left[\rho^{\otimes 2m} S_{2m} (O_1 \otimes I \otimes \cdots \otimes O_2 \otimes I \otimes \cdots \otimes I)\right] \\
		\mathrm{tr}(\rho^{m}\otimes \rho^m) & = \left[\mathrm{tr}\left(\rho^{\otimes m} S_{m} \right)\right]^2,
	\end{aligned}
\end{equation}
where the operator $S_{2m}$ is a shift operation over $2m$ copies quantum states, i.e., $S_{2m} \ket{b_1b_2\cdots b_{2m} }=\ket{b_2 \cdots b_{2m} b_1}$. Then Eq.~\eqref{eq:highorderexpec} can be estimated by
\begin{equation}
	\begin{aligned}
		\mathbb{E}\left\{\mathrm{tr}\left[\hat{\rho}_1 \otimes \hat{\rho}_2 \otimes \cdots \otimes \hat{\rho}_{2m} S_{2m} (O_1 \otimes I \otimes \cdots \otimes O_2 \otimes I \otimes \cdots \otimes I)\right]\right\} & = \mathrm{tr}(O_1 \rho^m O_2 \rho^m) \\
		\mathbb{E}\left[\mathrm{tr}\left(\hat{\rho}_1 \otimes \hat{\rho}_2 \otimes \cdots \otimes \hat{\rho}_{m} S_{m} \right)\right]                & = \mathrm{tr}(\rho^{m}),
	\end{aligned}
\end{equation}
where $\hat{\rho}_i$ $(i\in\{1,2,\dotsc, 2m\})$ are mutually indepedent.
We assume the quantum state prepared by the circuit is $\rho_{\mathrm{real}}=(1-\epsilon)\rho_{\mathrm{ideal}}+\epsilon \rho_{\mathrm{error}}$ and $\mathrm{tr}(\rho_\mathrm{ideal}\rho_\mathrm{error})=0$.
Then we can calculate the mitigated expectation in the presence of error. We denote $\sigma = \rho \otimes \rho$ and $O =S_2(O_1 \otimes O_2)$, then
\begin{equation}
	\begin{aligned}
		\sigma & = (1-\epsilon)^2 \rho_{\mathrm{ideal}} \otimes \rho_{\mathrm{ideal}} +\epsilon(1-\epsilon)(\rho_{\mathrm{ideal}} \otimes \rho_{\mathrm{error}}+\rho_{\mathrm{error}}\otimes \rho_{\mathrm{ideal}}) + \epsilon^2 \rho_{\mathrm{error}} \otimes \rho_{\mathrm{error}} \\
		       & := (1-\epsilon)^2 \sigma_{\mathrm{ideal}} + [1-(1-\epsilon)^2]\sigma_{\mathrm{error}}                             \\
		       & = (1-\epsilon^\prime) \sigma_{\mathrm{ideal}} + \epsilon^\prime \sigma_{\mathrm{error}},
	\end{aligned}
\end{equation}
and
\begin{equation}
	\begin{aligned}
		  & \frac{\mathrm{tr}(O_1 \rho^m O_2 \rho^m)}{\mathrm{tr}(\rho^{m}\otimes \rho^m)}                                               \\
		= & \frac{\mathrm{tr}(O \sigma^m)}{\mathrm{tr}(\sigma^m)}              \\
		= & \frac{(1-\epsilon^\prime)^m\mathrm{tr}(O\sigma_{\mathrm{ideal}})+\sigma^m \mathrm{tr}(O\sigma_{\mathrm{error}})}{(1-\epsilon^\prime)^m +\epsilon^\prime \mathrm{tr}(\sigma^m_{\mathrm{error}})}                       \\
		= & \mathrm{tr}(O\sigma_{\mathrm{ideal}}) +[\mathrm{tr}(O\sigma^m_{\mathrm{error}})-\mathrm{tr}(O\sigma_{\mathrm{ideal}})\mathrm{tr}(\sigma^m_{\mathrm{error}})](\epsilon^\prime)^m + \mathrm{O}((\epsilon^\prime)^{m+1}) \\
		= & \mathrm{tr}(O_1\rho O_2\rho) +[\mathrm{tr}(O\sigma^m_{\mathrm{error}})-\mathrm{tr}(O\sigma_{\mathrm{ideal}})\mathrm{tr}(\sigma^m_{\mathrm{error}})](\epsilon^\prime)^m + \mathrm{O}((\epsilon^\prime)^{m+1}),
	\end{aligned}
\end{equation}
which means the error can be reduced exponentially with $m$.

\bibliographystyle{apsrev4-1}

\bibliography{opensystemref}
\end{document}